\documentclass[a4paper,10pt]{article} 
\usepackage[%
  layoutsize={152.4mm,228.6mm},%
  layoutoffset=20mm,%
  twoside,%
  nomarginpar,%
  inner=25mm,%
  outer=10mm,%
  vmargin=17mm,%
  head=5mm,%
  headsep=2mm,%
  foot=7mm,%
  showcrop,%
  showframe,
]{geometry}
\usepackage[T1]{fontenc}
\usepackage{mathptmx}

\usepackage[compact]{titlesec}

\usepackage{fancyhdr}
\fancypagestyle{plain}{%
  \fancyhf{}
  \fancyfoot[CO,CE]{\thepage}
  \renewcommand{\headrulewidth}{0mm}
}

\usepackage{amsmath,amssymb,amsthm}
\usepackage{graphicx}
\usepackage{url}

\theoremstyle{plain}
\newtheorem{theorem}{Theorem}

\newtheorem{proposition}[theorem]{Proposition}
\newtheorem{lemma}[theorem]{Lemma}

\theoremstyle{definition}

\newtheorem{example}[theorem]{Example}
\newtheorem{problem}[theorem]{Problem}

\theoremstyle{remark}

\usepackage[colorlinks]{hyperref}

\newcommand{\numberofauthors}{}
\newcommand{\authorAfullname}{}
\newcommand{\authorAindexname}{}
\newcommand{\authorAaddress}{}

\newcommand{\authorBfullname}{}
\newcommand{\authorBindexname}{}
\newcommand{\authorBaddress}{}

\newcommand{\authorCfullname}{}

\newcommand{\authorCaddress}{}
\newcommand{\authorDfullname}{}

\newcommand{\authorDaddress}{}

\newcommand{\papertitle}{}

\newcommand{\processpaperdata}{%
  \setcounter{section}{0}
  \renewcommand{\leftmark}{%
    \ifcase\numberofauthors
      {}\or
      {\authorAfullname}\else
      {\authorAfullname{} et al.}
    \fi
  }
  \vspace*{3pt}
  \begin{center}{\bfseries\huge\papertitle}\end{center}
  \ifcase\numberofauthors
    {}\or
    {\begin{center}
      \begin{minipage}{50mm}\centering\authorAfullname\\\authorAaddress\end{minipage}
    \end{center}}\or
    {\begin{center}
      \begin{minipage}{50mm}\centering\authorAfullname\\\authorAaddress\end{minipage}
      \hspace{5mm}
      \begin{minipage}{50mm}\centering\authorBfullname\\\authorBaddress\end{minipage}
    \end{center}}\or
    {\begin{center}
      \begin{minipage}{50mm}\centering\authorAfullname\\\authorAaddress\end{minipage}
      \hspace{5mm}
      \begin{minipage}{50mm}\centering\authorBfullname\\\authorBaddress\end{minipage}
     \end{center}

     \begin{center}
      \begin{minipage}{50mm}\centering\authorCfullname\\\authorCaddress\end{minipage}
     \end{center}}\or
    {\begin{center}
      \begin{minipage}{50mm}\centering\authorAfullname\\\authorAaddress\end{minipage}
      \hspace{5mm}
      \begin{minipage}{50mm}\centering\authorBfullname\\\authorBaddress\end{minipage}
     \end{center}

     \begin{center}
      \begin{minipage}{50mm}\centering\authorCfullname\\\authorCaddress\end{minipage}
      \hspace{5mm}
      \begin{minipage}{50mm}\centering\authorDfullname\\\authorDaddress\end{minipage}
     \end{center}}\else
    {}
  \fi  
}

\renewenvironment{abstract}{\textbf{Abstract.}}{}

\begin{document}

\pagestyle{fancy}
\fancyhf{}
\fancyhead[RO,LE]{{\slshape Festschrift in Honor of Uwe Helmke}}
\fancyhead[LO,RE]{{\nouppercase{\slshape\leftmark}}}
\fancyfoot[CO,CE]{\thepage}
\renewcommand{\headrulewidth}{0mm}




\renewcommand{\numberofauthors}{2} 

\renewcommand{\authorAfullname}{J. Rosenthal}
\renewcommand{\authorAindexname}{Rosenthal, J.}
\renewcommand{\authorAaddress}{University of Zurich\\ 
Zurich, Switzerland\\ \texttt{rosenthal@math.uzh.ch}}

\renewcommand{\authorBfullname}{A.-L. Trautmann}
\renewcommand{\authorBindexname}{Trautmann, A.-L.}
\renewcommand{\authorBaddress}{University of Zurich\\ 
Zurich, Switzerland\\ \texttt{trautmann@math.uzh.ch}}

\renewcommand{\papertitle}{Decoding of Subspace Codes,\\
  a Problem of Schubert Calculus\\
  over Finite Fields}

\processpaperdata 

{





\newcommand{\rs}{\mathrm{rs}}
\newcommand{\C}{\mathbb{C}} 
\newcommand{\F}{\mathbb{F}}
\newcommand{\pn}[1]{{\mathbb P}^{#1}}
\newcommand{\G}{{\rm Grass}}
\newcommand{\Gq}{\mathrm{Grass}_q(k,n)}
\newcommand{\PG}{\mathcal{P}(q,n)}
\newcommand{\Vvs}{\mathcal{V}}
\newcommand{\Uvs}{\mathcal{U}}
\newcommand{\Rvs}{\mathcal{R}}
\newcommand{\Cvs}{\mathcal{C}}
\newcommand{\eqr}[1]{~\mbox{$(${\rm \ref{#1}}$)$}}


\begin{abstract}
  Schubert calculus provides algebraic tools to solve enumerative
  problems. There have been several applied problems in systems
  theory, linear algebra and physics which were studied by means of
  Schubert calculus.  The method is most powerful when the base field
  is algebraically closed.  In this article we first review some of
  the successes Schubert calculus had in the past. Then we show how
  the problem of decoding of subspace codes used in random network
  coding can be formulated as a problem in Schubert calculus. Since
  for this application the base field has to be assumed to be a finite
  field new techniques will have to be developed in the future.
\end{abstract}

\section{Introduction}

Hermann C\"asar Hannibal Schubert (1848-1911) is considered the
founder of enumerative geometry. He was a high school teacher in
Hamburg,  Germany.  He studied questions of the type:
Given four lines in projective three-space in general position, is there a line
intersecting all given ones. This question can then be generalized to:

\begin{problem} \label{SchPro} Given $N$ $k$-dimensional subspaces
  $\Uvs_i\subset \C^{k+m}$. Is there a subspace $\Vvs\subset \C^{k+m}$
  of complimentary dimension $m=\dim V$ such that
\begin{equation} \label{cond-2} 
\Vvs \bigcap \Uvs_i\neq \{ 0\}, \ i=1,\ldots,N.
  \end{equation}
\end{problem}

Using a symbolic calculus he then came up with the following surprising 
result~\cite{sc1886,sc1891}:

\begin{theorem} \label{Schub} 
In case the subspaces  $\Uvs_i\subset \C^{k+m}, i=1,\ldots,N$ are in general 
position and in case $N=km$ there exist exactly 
\begin{equation}                       \label{cond-1}
d(k,m)= \frac{1!2!\cdots (k-1)!(km)!}{m!(m+1)!\cdots(m+k-1)!}.
\end{equation}
different $m$ dimensional subspaces $\Vvs\subset \C^{k+m}$ which satisfy
the intersection condition\eqr{cond-2}.
\end{theorem}

Note that two-dimensional subspaces in $\C^4$ describe lines in
projective space $\mathbb{P}^3$ and Schubert hence claims in the case
of four lines in three-space in general position that there are exactly
$d(2,2)=2$ lines intersecting all four given lines.

Schubert used in the derivation of Theorem~\ref{Schub} Poncelet's
principle of preservation of numbers which was not considered a
theorem of mathematics at the time. For this reason Schubert's results
were not accepted by the mathematics community of the 19th century and
Hilbert devoted the 15th of his famous 24 problems to the question if
mathematicians can come up with rigorous techniques to prove or
disprove the claims of Dr. Schubert. A rigorous verification of
Theorem~\ref{Schub} was derived in the last century and we refer the
interested reader to the survey article~\cite{kl76} by Kleiman, where
the progress over time about Schubert calculus and the Hilbert problem
15 is described.

In the sequel we introduce the most important concepts from
Schubert calculus. 

Let $\F$ be an arbitrary field. Denote by $\G(k,n)=\G(k,\F^n)$ the
Grassmann variety consisting of all $k$-dimensional subspaces of the
vector space $\F^n$. $\G(k,n)$ can be embedded into projective space
using the Pl\"ucker embedding:
\begin{align*}
\varphi : \G(k,\F^n)  &\longrightarrow \mathbb{P}^{\binom{n}{k}-1} \\
\mathrm{span}(u_1,\ldots,u_k) &\longmapsto \F (u_1\wedge\ldots\wedge u_k).
\end{align*}

If one chooses a basis $\{e_1,\ldots,e_n\}$ of $\F^n$ and the corresponding
canonical basis of $\Lambda^k \F^n$
$$
\{ e_{i_1}\wedge \ldots \wedge 
 e_{i_k}\mid 1\leq i_{1} <\ldots <i_{k} \leq n\}
$$
then one has an induced map of the coordinates. If $U$ is a $k\times
n$ matrix whose row space $\rs (U)$ describes the subspace
$\Uvs:=\mathrm{span}(u_1,\ldots,u_k)$ and $U[i_1,\dots,i_k]$ denotes the submatrix of $U$ given by the columns $i_1,\dots,i_k$, then one readily verifies that
the Pl\"ucker embedding is given in terms of coordinates via:
$$
\rs (U) \longmapsto [\det(U[1,...,k]) : \det(U[1,...,k-1,k+1]) : ... : 
\det(U[n-k+1,...,n]).
$$
The $k\times k$ minors $\det(U[i_1,\ldots ,i_k])$ of the matrix $U$
are called the \emph{Pl\"ucker coordinates} of the subspace $\Uvs$.

The image of this embedding describes indeed a variety and the
defining equations are given by the so called ``shuffle relations''
(see e.g.~\cite{kl72,pr82}). The shuffle relations are a set of
quadratic equations in terms of the Pl\"ucker coordinates.

A flag ${\mathcal F}$ is a sequence of nested linear subspaces
$$
{\mathcal F}: \ \{ 0\}\subset V_{1}\subset V_{2} \subset\ldots
\subset V_{n}= {\F}^{n}
$$
having the property that $\dim V_{j}=j$ for $ j=1,\dots ,n$.

Denote by $\nu= ( \nu_{1},\ldots ,\nu_{k})$  an ordered index set
satisfying
$$
1\leq \nu_{1} <\ldots <\nu_{k} \leq n.
$$
For every  flag ${\mathcal F}$ 
one defines a
Schubert variety
\begin{equation}                              \label{schubvt}
\hspace{6mm}S(\nu;\mathcal{F}) := \{ W\in \G(m,\F^{n})\mid
\dim(W\bigcap V_{\nu_{i}}) \geq i \ \mbox{ for } i=1,\ldots ,k\}.
\end{equation}
The Schubert varieties are sub-varieties of the Grassmannian $\G(k,\F^n)$
and they contain a Zariski dense affine subset called 
Schubert cell and defined as:
\begin{equation}\label{schubcl}
C(\nu;\mathcal{F}):=\{W\in S(\nu; \mathcal{F})\mid \dim
(W\bigcap V_{\nu_{i}-1})=i-1 ; \mbox{ for } i=1,\ldots ,k\}.
\end{equation}

In terms of Pl\"ucker coordinates 
the defining equations of the Schubert variety  $S(\nu;\mathcal{F})$ are given
by the quadratic shuffle relations describing the Grassmann variety together
with a set of linear equations (see~\cite{kl72}). 

A fundamental question in Schubert calculus is now the following:
\begin{problem}
  Given two Schubert varieties $S(\nu;\mathcal{F})$ and
  $S(\tilde{\nu};\tilde{\mathcal{F}})$. Describe as explicitly as possible 
  the intersection variety
  $$
  S(\nu;\mathcal{F})\cap S(\tilde{\nu};\tilde{\mathcal{F})}.
  $$
\end{problem}

Schubert's Theorem~\ref{Schub} can actually also be formulated as an
intersection problem of Schubert varieties. For this note that
\begin{equation} \label{cond-3} 
\left\{  \Vvs\in\G(k,\F^{k+m})\mid  \Vvs \bigcap \Uvs_i\neq \{ 0\}\right\}
  \end{equation}
describes a Schubert variety with regard to some flag and the theorem then
states that in the intersection of $N$ Schubert varieties of above type
one finds $d(k,m)$ $m$-dimensional subspaces as solutions in general.

In the case of an algebraically closed field one has rather precise
information about this intersection variety. Topologically the
intersection variety turns out to be a union of Schubert varieties of
lower dimension and the multiplicities are governed by the
Littlewood--Richardson rule~\cite{he95}. When the field is not
algebraically closed much less is known. There has been work done over
the real numbers by Frank Sottile~\cite{so97a,so00a}. Over general
fields very little is known and we will show in this article that the
decoding of subspace codes can be viewed as a Schubert calculus
problem over some finite field. The following example illustrates
the concepts.

\begin{example}
As a base field we take $\F=\F_2=\{0,1\}$ the binary field.
Consider the Grassmannian $\G_2(2,\F^4)$ representing all lines
in projective three-space $\mathbb{P}^3$. We would like to study
Schubert's question in this situation: Given four lines in three-space,
is there always a line intersecting all four given ones. Clearly there are
many situations where the answer is affirmative, e.g. when the lines 
already intersect in some point. In general this is however not the 
case as we now demonstrate. Consider the following four lines in $\mathbb{P}^3$
represented as row spaces of the following four matrices:
$$
\left[\begin{array}{cccc}
    1 & 0 & 0 & 0 \\
    0 & 1 & 0 & 0
  \end{array}\right],
\left[\begin{array}{cccc}
    0 & 0 & 1 & 0 \\
    0 & 0 & 0 & 1
  \end{array}\right],
\left[\begin{array}{cccc}
    1 & 1 & 0 & 1 \\
    0 & 1 & 1 & 1
  \end{array}\right],
\left[\begin{array}{cccc}
    1 & 1 & 1 & 0 \\
    1 & 0 & 1 & 1
  \end{array}\right].
$$

We claim that there exists no line in projective three-space
$\mathbb{P}^3$, i.e. no two-dimensional subspace in $\G_2(2,\F^4)$
intersecting all four given subspaces non-trivially. 

$\G_2(2,\F^4)$ is embedded in $\mathbb{P}^5$ via the Pl\"ucker embedding.
Denote by 
$$
u_{i,j}:=\det U[i,j], 1\leq i<j\leq 4
$$
the Pl\"ucker coordinates of some subspace $\Uvs\in \G_2(2,\F^4)$.
The four lines impose the linear constraints:
\begin{eqnarray*}
u_{3,4}&=&0,\\
 u_{1,2}&=&0,  \\
u_{1,2}+u_{1,4}+u_{2,3}+u_{2,4}+u_{3,4}&=&0,\\
u_{1,2}+u_{1,3}+u_{1,4}+u_{2,3}+u_{3,4}&=&0.
\end{eqnarray*}
The points in $\mathbb{P}^5$ representing the image of  $\G_2(2,\F^4)$
are described by one quadratic equation (shuffle relation):
$$
u_{1,2}u_{3,4}+u_{1,3}u_{2,4}+u_{1,4}u_{2,3}=0.
$$
Solving the 5 equations in the 6 unknowns results in one
quadratic equation:
$$
(u_{1,4})^2+u_{1,4}u_{2,3}+(u_{2,3})^2=0
$$
which has no solutions over $\F_2$ in $\mathbb{P}^5$. Note that there 
are exactly $d(2,2)=2$ solutions over the algebraic closure 
as predicted by Schubert.
\end{example}

Readers who want to know more on the subject of Schubert calculus will
find material in the survey article~\cite{kl72}.

The paper is structured as follows: In Section~\ref{SchuAppl} we
present results which were derived by Schubert calculus. In
Section~\ref{NetCode} we introduce the main topic of this paper,
namely subspace codes used in random network coding. In
Section~\ref{ListDec} we show that list decoding of random network
codes is a problem of Schubert calculus over some finite field.

\section{Results in Systems Theory and Linear Algebra 
Derived by Means of Schubert Calculus}      \label{SchuAppl}

In the past Schubert calculus has been a very powerful tool
for several problem areas in the applied sciences. In this section 
we review two such problem areas and we show to what extend Schubert
calculus led to strong existence results and better understanding.

\subsection*{The pole placement problem}
One of the most prominent problems in mathematical systems theory
has been the pole placement problem. In the static situation 
the problem can be described as follows:
Consider a discrete linear system 
\begin{equation}                     \label{system}
  x(t+1)=Ax(t)+Bu(t),\ \   y(t)=Cx(t)
\end{equation}
described by matrices $A,B,C$ having size 
$n\times n$, $n\times m$ and $p\times n$ respectively. Consider 
a monic polynomial
$$
\varphi(s):=s^n+a_{n-1}s^{n-1}+\cdots+a_1s+a_0\in \F[s]
$$
of degree $n$ having coefficients in the base field $\F$. In its simplest
version the pole placement problem asks for the existence of a 
feedback law $u(t)=Ky(t)$ such that the resulting closed loop
system 
\begin{equation}
 x(t+1)=\left( A+BKC\right)x(t)
\end{equation}
has characteristic polynomial $\varphi(s)$. 

At first glance this problem looks like a problem from matrix theory
whose solution can be derived by means of linear algebra.
Surprisingly, the problem is highly nonlinear and closely related to
Schubert's Problem~\ref{SchPro}. This geometric connection was first
realized in an interesting paper by Brockett and Byrnes~\cite{br81}
who showed that over the complex numbers arbitrary pole placement is
generically possible as soon as $n\leq mp$ and in case that the
McMillan degree $n$ is equal to $mp$ then there are exactly $d(m,p)$
static feedback laws resulting in the closed loop characteristic
polynomial $\varphi(s)$.  The interested reader will find more details
in a survey article by Byrnes~\cite{by89}.

The geometric insight one gained from the Grassmannian point of view 
was also helpful for deriving pole placement results over other
base field. Over the reals the most striking result was obtained
by A.~Wang in~\cite{wa92} where it was shown that arbitrary pole
placement is possible with real compensators as soon as $n<mp$. 
Over a finite field some preliminary results were obtained by Gorla
and the first author in~\cite{go10}.

U. Helmke  in collaboration with X. Wang and the first author have been
studying the pole placement problem in the situation when 
symmetries are involved~\cite{he06}. This problem then leads 
to a Schubert type problem in the Lagrangian Grassmannian.

\subsection*{Sums of Hermitian matrices}

Given Hermitian matrices $A_{1},\ldots ,A_{r}\in\C^{n\times
n}$ each with a fixed spectrum 
\begin{equation}
\lambda_{1}(A_{l})\geq
\ldots\geq \lambda_{n}(A_{l}),\ \  l=1,\ldots ,r
\end{equation}
 and arbitrary
else. Is it possible to find then linear inequalities which
describe the possible spectrum of the Hermitian matrix 
$$
A_{r+1}:=A_{1}+\cdots +A_{r}?
$$

Questions of this type have a long history in operator theory
and linear algebra. For example H. Weyl derived in 1912 the following
famous inequality for any set of indices $1\leq i,j\leq n$ with 
$1\leq i+j-1\leq n$:
\begin{equation}
\lambda_{i+j-1}(A_1+A_2)\leq \lambda_i(A_1)+\lambda_j(A_2).
\end{equation}

In collaboration with U. Helmke the first author extended 
work by Johnson~\cite{jo79} and Thompson~\cite{th71,th89} to derive
a large set of eigenvalue inequalities. This was achieved through the 
use of Schubert calculus and we will say more in a moment.
The obtained inequalities 
included in special cases
not only the inequalities by H. Weyl but also the more extensive
inequalities from Lidskii and Freede Thompson~\cite{th71}.

In order to make the connection to Schubert calculus we 
follow~\cite{he95} and denote with $v_{1l},\ldots , v_{nl}$
the set of orthogonal eigenvectors of the Hermitian operator $A_l$,
$l=1,\ldots,r+1$. 

Using these ordered set of eigenvectors one constructs 
for  each Hermitian matrix $A_{l}$ the flag:
\begin{equation}
{\cal F}_{l}:\hspace{4mm}
\{ 0\}\subset V_{1l}\subset V_{2l}\subset\ldots \subset
V_{nl}= \C^{n}
\end{equation}
defined through the property:
\begin{equation}                \label{eig}
V_{ml}:={\rm span}(v_{1l},\ldots , v_{ml})\; \; m=1,\dots ,n.
\end{equation}

The connection to Schubert calculus is now established 
by the following result as it can be found in~\cite{he95}. 
The theorem generalizes earlier results by Freede and Thompson~\cite{th71}.

\begin{theorem}                \label{lemma2}
Let $A_{1},\ldots ,A_{r}$ be complex Hermitian $n\times n$
matrices and denote with ${\mathcal F}_{1},\ldots ,{\mathcal F}_{r+1}$
the corresponding flags of eigenspaces defined by\eqr{eig}.
Assume $A_{r+1}=A_{1}+\cdots +A_{r}$.
 and let
$\underline{i}_{l}= (i_{1l},\ldots ,i_{kl})$ be
$r+1$ sequences of integers satisfying
\begin{equation}
1\leq i_{1l} <\ldots <i_{kl} \leq n,\hspace{2mm} l=1,\ldots,r+1.
\end{equation}
Suppose the intersection of the $r+1$ Schubert subvarieties
of $\G(k,\C^{n})$ is nonempty, i.e.:
\begin{equation}                  \label{inter}
S(\underline{i}_{1};{\cal F}_{1})\bigcap\ldots \bigcap
S(\underline{i}_{r+1};{\cal F}_{r+1})\neq \emptyset.
\end{equation}
Then the following matrix eigenvalue inequalities hold:
\begin{equation}                   \label{ineq1}
\sum_{j=1}^{k}\lambda_{n-i_{j,r+1}+1}(A_{1}+\cdots +A_{r})
\geq \sum_{l=1}^{r}\sum_{j=1}^{k}\lambda_{i_{jl}}(A_{l})
\end{equation}
\begin{equation}                   \label{ineq2}
\sum_{j=1}^{k}\lambda_{i_{j,r+1}}(A_{1}+\cdots +A_{r})
\leq \sum_{l=1}^{r}\sum_{j=1}^{k}\lambda_{n-i_{jl}+1}(A_{l}).
\end{equation}
\end{theorem}

In 1998 Klyachko could show that the inequalities coming from Schubert
calculus as described in Theorem~\ref{lemma2} are not only necessary
but that they describe a Polytope of all possible inequalities.  The
interested reader will find Klyachko's result as well as much more in
the survey article by Fulton~\cite{fu00}.

A priori classical Schubert calculus provides very strong existence results.
It is a different matter to derive effective numerical algorithms 
to compute the subspaces which satisfy the different Schubert conditions.
For this reason Huber, Sottile  and Sturmfels~\cite{hu98} developed 
effective numerical algorithms over the reals. As we will demonstrate 
in the next sections it would be very desirable to have effective numerical
algorithms also in the case of Schubert type problems defined over some
finite field.

\section{Random Network Coding}            \label{NetCode}

In network coding one is looking at the transmission of information
through a network with possibly several senders and several receivers.
A lot of real-life applications of network coding can be found, e.g.
data streaming over the Internet, where a source wants to send the
same information to many receivers at the same time.

The network channel is represented by a directed graph with three
different types of vertices, namely \emph{sources}, i.e. vertices with
no incoming edges, \emph{sinks}, i.e. vertices with no outgoing edges,
and \emph{inner nodes}, i.e. vertices with incoming and outgoing
edges.  One assumes that at least one source and one sink exist.
Under \emph{linear} network coding the inner nodes are allowed to
forward linear combinations of the incoming information vectors. The
use of linear network coding possibly improves the transmission rate
in comparison to just forwarding information at the inner nodes
\cite{ah00}. This can be illustrated in the example of the butterfly
network:
%
%
\begin{figure}[h]
 \centering
\scalebox{0.15}{\includegraphics{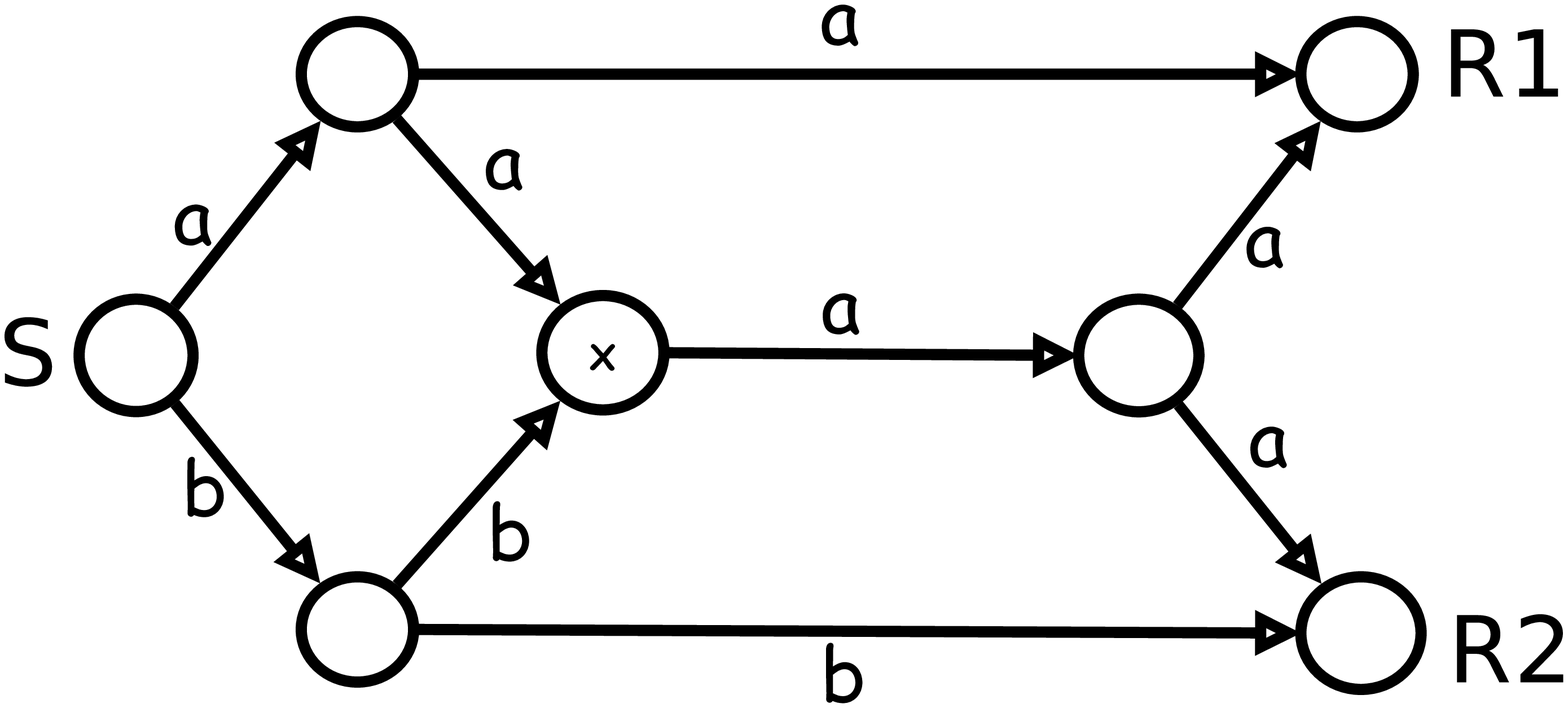}} \hspace{1cm}
\scalebox{0.15}{\includegraphics{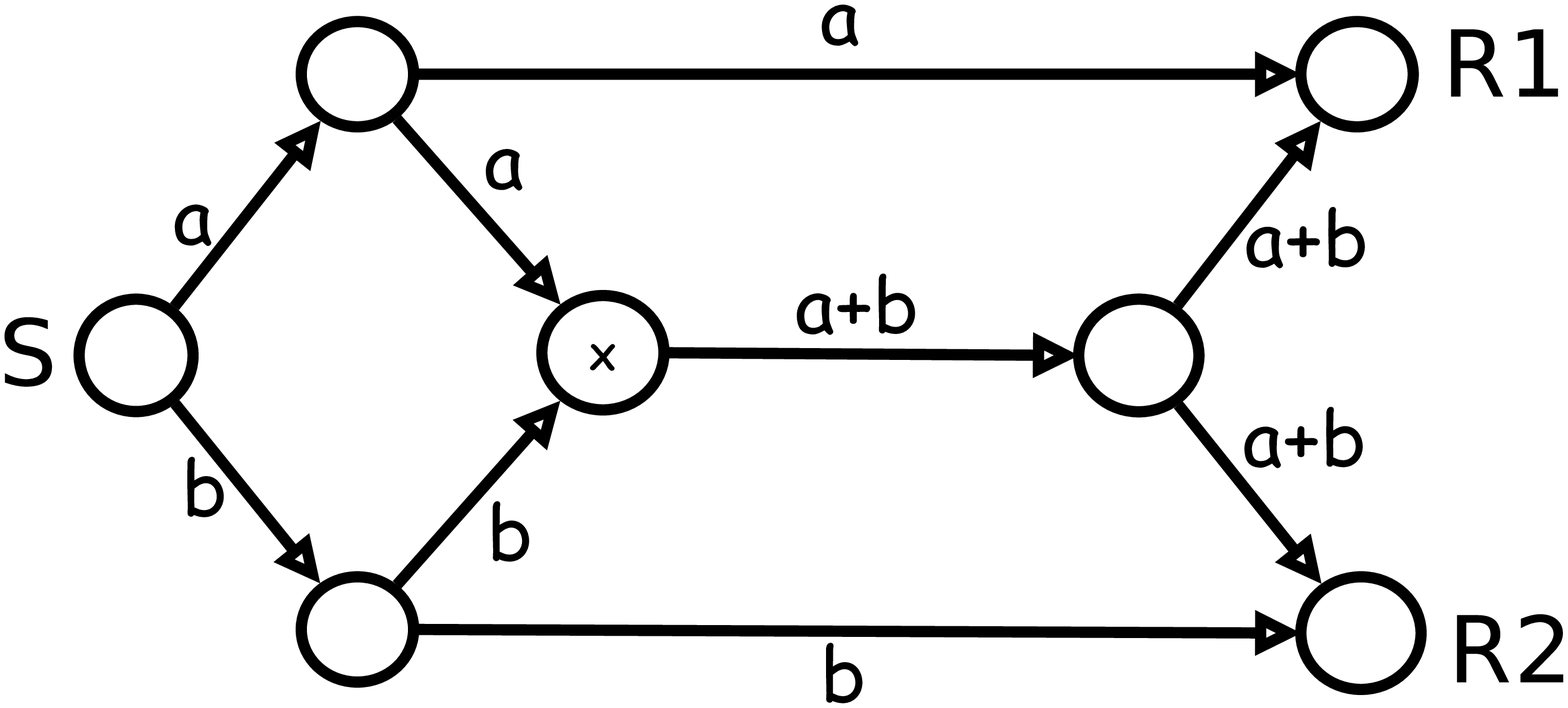}}
\caption{The butterfly network under the forwarding and the network coding model.}
\end{figure}
The source $S$ wants to send the same information, $a$ and $b$, to
both receivers $R1$ and $R2$. Under forwarding every inner node
forwards the incoming information and thus has to decide on either $a$
or $b$ (in this example on $a$) at the bottleneck vertex, marked above
by x. Thus, $R1$ does not receive $b$. With linear network coding we
allow the bottleneck vertex to send the sum of the two incoming
informations, which allows both receivers to recover both $a$ and $b$
with a simple operation.

In this linear network coding setting, when the topology of the
underlying network is unknown or time-varying, one speaks of
\emph{random} (linear) network coding. This setting was first studied
in \cite{ho03a2} and a mathematical model was introduced in
\cite{ko08}, where the authors showed that it makes sense to use
vector spaces instead of vectors over a finite field $\F_{q}$ as
codewords. In this model the source injects a basis of the respective
codeword into the network and the inner nodes forward a random linear
combination of their incoming vectors. Therefore, each sink receives a
linear combinations of the original vectors, which span the same
vector space as the sent vectors, if no errors occurred during
transmission.

In coding practice the base field is a finite field $\mathbb{F}_q$
having $q$ elements, where $q$ is a prime power. $\F_{q}^{\times} :=
\F_{q}\setminus \{0\}$ will denote the set of all invertible elements of
$\F_{q}$.  We will call the set of all subspaces of $\F_q^n$ the
projective geometry of $\F_q^n$, denoted by $\PG$, and denote the Grassmannian $\mathrm{Grass}(k,\F_q^n)$ by
$\Gq$.

There are two types of errors that may occur during transmission, a
decrease in dimension which is called an \emph{erasure} and an
increase in dimension, called an \emph{insertion}. Assume $\Uvs\in
\PG$ was sent and erasures and insertions occurred during
transmission, then the received word is of the type
\[
\Rvs = \bar{\Uvs} \oplus \mathcal{E}
\]
where $\bar{\Uvs} $ is a subspace of $\Uvs$ and 
$\mathcal{E}\in \PG$ is the error space. 
A random network coding channel in which both 
insertions and erasures can happen is called an \emph{operator channel}.

In order to have a notion of decoding capability of some code
a good metric is required on the set $\PG$:
The \emph{subspace distance}  is a metric on $\PG$
given by
\begin{align*}
  d_S(\mathcal{U},\mathcal{V}) =& \dim(\Uvs + \Vvs) - \dim(\mathcal{U}\cap
  \mathcal{V})\\
  =& \dim(\Uvs) + \dim(\Vvs) - 2\dim(\mathcal{U}\cap
  \mathcal{V})
\end{align*}
for any $\mathcal{U},\mathcal{V} \in \PG$. Another metric on 
$\PG$ is the \emph{injection distance}, defined as
\begin{align*}
  d_I(\mathcal{U},\mathcal{V}) =& 
\max\{\dim(\Uvs),\dim(\Vvs)\} - \dim(\mathcal{U}\cap
  \mathcal{V}) .
\end{align*}
Note, that for $\Uvs, \Vvs \in \Gq$ it holds that $d_{S}(\Uvs,
\Vvs)=2d_{I}(\Uvs,\Vvs)$.  A \textit{subspace code} $\mathcal{C}$ is
simply a subset of $\PG$. If $\Cvs \subseteq \Gq$, we call it a
\emph{constant dimension code}. The minimum distance of a subspace code is defined in the
usual way.

Different constructions of subspace codes have been studied, e.g. in \cite{et08u,et08p,ko08p,ko08,ma08p,ro12,si08a,tr11p}. Some facts on isometry classes and automorphisms of these codes can be found in \cite{tr12}.

 The set of all invertible $n\times n$-matrices
with entries in $\F_{q}$, called the general linear group, is denoted
by $GL_{n}$. Moreover, the set of all $k\times n$-matrices over $\F_q$
is denoted by $\F_q^{k\times n}$.

Let $U\in \F_q^{k\times n}$ be a matrix of rank $k$ and
\[
\mathcal{U}=\rs (U):= \text{row space}(U)\in \Gq.
\]
One can notice that the row space is invariant under $GL_k$-multiplication from
the left, i.e. for any $T\in GL_k$
\[
\mathcal{U}=\rs(U)= \rs(T U).
\]
Thus, there are several matrices that represent a given subspace. A
unique representative of these matrices is the one in reduced row
echelon form.
Any $k\times n$-matrix can be transformed into reduced row echelon
form by a $T\in GL_k$.

Given $U\in \F_q^{k\times n}$ of rank $k$,
$\mathcal{U}\in \Gq$ its row space and $A\in GL_n$, we
define
\[
\mathcal{U} A:=\rs(U A).
\]

Let $U,V\in \F_q^{k\times n}$ be matrices such that
$\rs(U)=\rs(V)$. Then one readily verifies that $\rs(U A)=\rs(V A)$
for any $A\in GL_n$, hence the operation is well defined.

\subsection*{Decoding subspace codes}

Given a subspace code $\Cvs \subseteq \PG$ and a received codeword
$\Rvs \in \PG$, a \emph{maximum likelihood decoder} decodes to a
codeword $\Uvs \in \Cvs$ that maximizes the probability
\[    P(\Rvs \mbox{ received} \mid \Uvs \mbox{ sent}) \]
over all $\Uvs \in \Cvs$. 

A \emph{minimum distance decoder} chooses the closest codeword to the
received word with respect to the subspace or injection distance. Let
us assume that both the erasure and the insertion probability is
less than some fixed $\epsilon$. Then over an operator channel where the
insertion probability is equal to the erasure probability, maximum
likelihood decoding is equivalent to minimum distance decoding with
respect to the subspace distance while in an adversarial model it is
equivalent to minimum distance decoding with respect to the injection
distance \cite{si09}.

Assume the minimum (injection) distance of $\Cvs$ is $d$, then if
there exists $\Uvs \in \Cvs$ with $d_{I}(\Rvs, \Uvs) \leq
\frac{d-1}{2}$, then $\Uvs$ is the unique closest codeword and the
minimum distance decoder will always decode to $\Uvs$.

Note, that a minimum subspace distance decoder is equivalent to a
minimum injection distance decoder when $\Cvs$ is a constant dimension
code. Since we will investigate constant dimension codes in the
remainder of this paper we will always use the injection distance. All
results can then be carried over to the subspace distance.

A very important concept in coding theory is the problem of {\em list 
decoding}. It is the goal of list decoding to come up with an algorithm
which allows one to compute all code words which are within some distance
of some received subspace.

For some $\Uvs \in \PG$ we denote the
ball of radius $e$ with center $\Uvs$ in $\PG$ by $B_{e}(\Uvs)$. If we
want to describe the same ball inside $\Gq$ we denote it by
$B_{e}^{k}(\Uvs)$.  Note that for a constant dimension code the ball
$B_{e}^{k}(\Uvs)$ is nothing else than some Schubert variety of $\Gq$.

Given a subspace code $\Cvs \subseteq \PG$ and a received codeword
$\Rvs \in \PG$, a \emph{list decoder with error bound e} outputs a
list of codewords $\Uvs_{1},\dots, \Uvs_{m} \in \Cvs$ whose injection
(resp. subspace) distance from $\Rvs$ is at most $e$. In other words,
the list is equal to the set
\[B_{e}(\Rvs) \cap \Cvs .\]
If $\Cvs$ is a constant dimension code, then the output of the list decoder becomes
$B^{k}_{e}(\Rvs) \cap \Cvs $.


\section{List Decoding in Pl\"ucker Coordinates}  \label{ListDec}

As already mentioned before 
 the balls of radius $t$ (with respect to the subspace
distance) around some $\Uvs \in \Gq $ forms a Schubert 
variety over a finite field. In terms of  Pl\"ucker
coordinates it is possible to give explicit equations.
 For it we need the following monomial order:
\[(i_{1},...,i_k) > (j_{1},...,j_{k}) \iff 
\exists N\in \mathbb{N}_{\geq 0} : i_{l}=j_{l} \;\forall l<N 
\textnormal{ and } i_{N} > j_{N}   .\]
It is easy to compute the balls in the following special case.

\begin{proposition}
  Define $\Uvs_{0}:=\rs [\begin{array}{cc}I_{k\times k} &0_{k\times
      n-k} \end{array}]$. Then
\begin{align*}
B^k_{t}(\Uvs_{0}) = \{&\Vvs \in \Gq  \mid \varphi(\Vvs) = 
[\mu_{1,\dots,k}: \dots : \mu_{n-k+1,\dots,n}]
, \\
&\mu_{i_1,...,i_{k}} = 0  \;\forall (i_{1},...,i_{k}) 
\not \leq  (t+1,\dots,k,n-t+1,...,n)  \}
\end{align*}
\end{proposition}

\begin{proof}
For $\Vvs$ to be inside the ball it has to hold that
\begin{align*}
d_{I}(\Uvs_{0}, \Vvs) &\leq t \\
\iff k - \dim(\Uvs_{0} \cap \Vvs) & \leq  t \\
\iff \dim(\Uvs_{0} \cap \Vvs) & \geq k-t
\end{align*}
i.e. $k-t$ many of the unit vectors $e_{1},..., e_{k}$ have to be
elements of $\Vvs$. 
Therefore $\varphi(\Vvs)$ has to fulfill
\[
\mu_{i_1,...,i_{k}} = 0  \textnormal{ if } 
(i_{1},...,i_{k}) \not \leq  (t+1,...,k,n-t+1,...,n) .
\]
 \end{proof}

 With the knowledge of $B^k_{t}(\Uvs_0)$ we can also express
 $B^k_{t}(\Uvs)$ for any $\Uvs \in \Gq $. For this note, that for any
 $\Uvs \in \Gq $ there exists an $A\in GL_n$ such that $\Uvs_0 A= \Uvs$.
 Moreover,
\[B^k_{t}(\Uvs_0 A) = B^k_{t}(\Uvs_0) A .\]

For simplifying the computations we define $\varphi$ 
on $GL_n$, where we denote by $A_{i_{1},\dots, i_k}$ the submatrix of $A$ that consists of the rows $i_{1}, \dots, i_{k}$:
\begin{align*}
\varphi : GL_n &\longrightarrow GL_{\binom{n}{k}} \\
 A & \longmapsto \left(\begin{array}{cccccc}
\det A_{1,\dots, k}[1, \dots, k] & \dots & \det A_{1,\dots, k}[n-k+1 ,\dots, n]\\
\vdots & & \vdots \\
\det A_{n-k+1,\dots, n}[1, \dots, k] & \dots & \det A_{n-k+1,\dots, n}[n-k+1 ,\dots, n]
                   \end{array}
 \right) 
\end{align*}


\begin{lemma}
Let $\Uvs \in \Gq $ and $A\in GL_{n}$. It holds that
\[\varphi(\Uvs A) = \varphi(\Uvs) \varphi(A).\]
\end{lemma}

\begin{theorem}\label{thm5}
Let $\Uvs= \Uvs_{0}A \in \Gq $. Then
\begin{multline*}
B^k_{t}(\Uvs) =B^k_{t}(\Uvs_{0} A) \\
=\{\Vvs  \in \Gq  \mid \varphi(\Vvs) \varphi(A^{-1})
 = [\mu_{0,\dots,k-1}: 
\dots : \mu_{n-k+1,\dots,n}],\\
 \mu_{i_{1},\dots,i_{k}} = 0 
\;\forall (i_{1},\dots,i_{k}) \not \leq  (t+1,\dots,k,n-t+1,...,n)\}.
\end{multline*}
\end{theorem}

There are always several choices for $A\in GL_{n}$ such that
$\Uvs_{0}A=\Uvs$. Since $GL_{{\binom{n}{k}}}$ is very large we try to
choose $A$ as simple as possible. We will now explain one such
construction:
\begin{enumerate}
\item The first $k$ rows of $A$ are equal to the matrix representation
  $U$ of $\Uvs$.
\item Find the pivot columns of $U$ (assume that $U$ is in reduced row echelon form).
\item Fill up the respective columns of $A$ with zeros in the lower $n-k$ rows.
\item Fill up the remaining submatrix of size $n-k\times n-k$ with an identity matrix.
\end{enumerate}
Then the inverse of $A$ can be computed as follows:
\begin{enumerate}
\item Find a permutation $\sigma \in S_{n}$ that permutes the columns
  of $A$ such that
\[\sigma(A)= \left(\begin{array}{ccc} I_k & U'' \\ 0 & I_{n-k} \end{array}\right) .\]
\item Then the inverse of that matrix is
\[\sigma(A)^{-1}= \left(\begin{array}{ccc} I_k & -U'' \\ 0 & I_{n-k} \end{array}\right).\]
\item Apply $\sigma$ on the rows of $\sigma(A)^{-1}$. The result is
  $A^{-1}$.  One can easily see this if one represents $\sigma$ by a
  matrix $S$. Then one gets $(SA)^{-1}S=A^{-1}S^{-1}S=A^{-1}$.
\end{enumerate}

\begin{example}
In $\mathcal{G}_{2}(2,4)$ we want to find 
\[B_{1}^{2}\left(\Uvs \right) = \{\Vvs  
\in \mathcal{G}_{2}(2,4) \mid\Vvs\cap \Uvs=1\}\]
for 
\[\Uvs =\rs (U)= \rs \left[
\begin{array}{cccc} 1&0&0&0\\ 0&0&1&1 \end{array}\right] .
\]
We find the pivot columns $U[1,3]$ and build 
\[A= \left(\begin{array}{cccc} 1 & 0&0&0 \\ 
0&0&1&1 \\ 0&1&0&0 \\ 0&0&0&1 \end{array}\right).
\]
Then we find the column permutation $\sigma=(23)$ such that
\[\sigma(A)= \left(\begin{array}{cccc} 1 & 0&0&0 \\ 
0&1&0&1 \\ 0&0&1&0 \\ 0&0&0&1 \end{array}\right).
\]
Now we can easily invert as described above and see that 
$\sigma(A)^{-1}=\sigma(A)$. We apply $\sigma$ on the rows and get
\[
A^{-1}= \left(\begin{array}{cccc} 1 & 0&0&0 \\ 
0&0&1&0 \\0&1&0&1 \\  0&0&0&1 \end{array}\right).
\]
Then 
\[
\varphi(A^{-1})= \left(\begin{array}{cccccc}
0& 1&0 & 0&0&0 \\ 1&0&1&0&0&0 \\0&0&1&0&0&0 
\\  0&0&0&1&0&1\\ 0&0&0&0&0&1 \\ 0&0&0&0&1&0 
\end{array}\right).
\]
From Theorem \ref{thm5} we know that
\begin{multline*}
B^2_{1}(\Uvs) = \{\Vvs  \in \mathcal{G}_{2}(2,4) 
\mid \varphi(\Vvs) \varphi(A^{-1}) = [\mu_{1,2}:\dots:\mu_{3,4}],\\
 \mu_{i_1,i_{2}} = 0 
 \;\forall (i_{1},i_{2}) \not \leq  (2,4)  \}\\
=\{\Vvs  \in \mathcal{G}_{2}(2,4) \mid \varphi(\Vvs) \varphi(A^{-1})
 = [\mu_{1,2}:\mu_{1,3}: \dots : \mu_{3,4}], \mu_{3,4} = 0  \}
\end{multline*}
Now let $\varphi(\Vvs)  = [\nu_{1,2}:\nu_{1,3}: 
\nu_{1,4}: \nu_{2,3}: \nu_{2,4}: \nu_{3,4}]$, then 
\[\varphi(\Vvs) \varphi(A^{-1}) = [\nu_{1,3}:\nu_{1,2}: 
\nu_{1,3}+\nu_{1,4}: \nu_{2,3}: \nu_{3,4}: \nu_{2,3}+\nu_{2,4}]
\]
and hence
\begin{multline*}
B^2_{1}(\Uvs)\\
 =\{\Vvs  \in \mathcal{G}_{2}(2,4) \mid \varphi(\Vvs)
 = [\nu_{1,2}:\nu_{1,3}: \nu_{1,4}: \nu_{2,3}: \nu_{2,4}: 
\nu_{3,4}], \nu_{2,3}+\nu_{2,4} = 0  \}\\
=\{\Vvs  \in \mathcal{G}_{2}(2,4) \mid \varphi(\Vvs)  
=[\nu_{1,2}:\nu_{1,3}: \nu_{1,4}: \nu_{2,3}: \nu_{2,4}: 
\nu_{3,4}], \nu_{2,3}=\nu_{2,4}  \}.
\end{multline*}
Note, that we do not have to compute the whole matrix
$\varphi(A^{-1})$ since in this case we only need the last column of
it to find the equations that define $B^2_{1}(\Uvs)$.
\end{example}

\section{Conclusion}\label{conclusion}

The article explains the importance of Schubert calculus in various
areas of systems theory and linear algebra. The strongest results 
in Schubert calculus require that the base field is algebraically closed.
The problem of list decoding subspace codes is a 
problem of Schubert calculus where the underlying field is a finite field.
It will be a topic of future research to come up with efficient algorithms
to tackle this problem computationally.

Many of the results we describe in this paper were derived by the 
first author in collaboration with Uwe Helmke. This collaboration was 
always very stimulating and the first author would like to thank Uwe Helmke
for this continuing collaboration.

\section*{Acknowledgments}

Research partially supported by Swiss National Science Foundation
Project no. 138080.

}

\nocite{he95,ro12,he06,gu04}
\bibliographystyle{abbrv}
\bibliography{/home/b/trautman/huge}

\end{document}